\documentclass[11pt]{article}
\pdfoutput=1
\usepackage{amsmath}
\usepackage{amssymb}
\usepackage{amsthm}
\usepackage{authblk}
\usepackage{fullpage}

\usepackage[mode=buildnew,subpreambles=true]{standalone}
\usepackage{tikz}

\graphicspath{{./includes/}}

\title{A Critique of Keum-Bae Cho's Proof that $\pe \subsetneq \np$}

\author{Benjamin Carleton\thanks{Supported in part by NSF grant CCF-2006496.}}
\author{Michael C. Chavrimootoo\protect\footnotemark[1]}
\author{Conor Taliancich\protect\footnotemark[1]}
\affil{Department of Computer Science\\University of Rochester\\Rochester, NY 14627, USA}

\date{April 4, 2021}

\newcommand{\pe}{\mbox{\rm P}}
\newcommand{\np}{\mbox{\rm NP}}
\newcommand{\npctext}{NP-complete}
\newcommand{\sat}{\mbox{\rm SAT}}
\newcommand{\threecnftext}{3-CNF}
\newcommand{\threecnf}{\mbox{\rm \threecnftext}}
\newcommand{\threesattext}{3-SAT}
\newcommand{\threesat}{\mbox{\rm \threesattext}}
\newcommand{\twosat}{\mbox{\rm 2-SAT}}
\newcommand{\hornsat}{\mbox{\rm Horn-SAT}}
\newcommand{\true}{\mbox{\it true}}
\newcommand{\false}{\mbox{\it false}}

\newtheorem{theorem}{Theorem}
\newtheorem{lemma}{Lemma}

\begin{document}

\maketitle

\begin{abstract}
In this paper we critique Keum-Bae Cho's proof that $\pe \subsetneq \np$. This proof relates instances of \threesat{} to indistinguishable binomial decision trees and claims that no polynomial-time algorithm can solve \threesat{} instances represented by these trees. We argue that their proof fails to justify a crucial step, and so the proof does not establish that $\pe \subsetneq \np$.
\end{abstract}

\section{Introduction} \label{sec:intro}
The famous \pe{} vs.~\np{} problem is concerned with establishing the relationship between \pe, the class of languages decided by deterministic polynomial-time Turing Machines, and \np, the class of languages decided by nondeterministic polynomial-time Turing Machines. The class \pe{} is trivially known to be contained in \np. However, it remains unknown whether the two classes are equal. One way to approach this problem is to consider \npctext{} problems---sometimes called the ``hardest'' problems in \np. Showing that an \npctext{} problem can be decided by a polynomial-time algorithm would prove that $\pe = \np$, and showing that no polynomial-time algorithm can decide an \npctext{} problem would show that $\pe \neq \np$. An example of an \npctext{} problem is \threesat.

In the paper ``Indistinguishable binomial decision tree of 3-SAT\@: Proof of $\pe \subsetneq \np$''~\cite{cho:indistinguishable}, Keum-Bae Cho claims that \threesat{} has no polynomial-time algorithm. They do so by purportedly providing a superpolynomial lower bound on \threesat. In this critique, we show that their proof is erroneous and does not resolve the \pe{} vs.~\np{} problem.

In Section~\ref{sec:prelim} we introduce the definitions and concepts needed to understand Keum-Bae Cho's arguments. We explain Cho's argument in Section~\ref{sec:summary} and present our critique in Section~\ref{sec:critique}. We provide concluding remarks in Section~\ref{sec:conclusion}.

\section{Preliminaries} \label{sec:prelim}
\subsection{\threecnftext{} (see \cite{hop-ull:b:automata})}
If $x$ is a Boolean variable, then $x$ and $\overline{x}$ are called \textit{literals}. A \textit{clause} is a disjunction of one or more literals. For example, $x_1 \lor \overline{x_2} \lor x_3$ is a clause with variables $x_1$, $x_2$, and $x_3$, and literals $x_1$, $\overline{x_2}$, and $x_3$. A conjunction of one or more clauses is called a \textit{CNF formula}. For example, $(x_1 \lor \overline{x_2} \lor x_3 \lor x_4) \land (x_2 \lor x_4) \land (\overline{x_1})$ is a CNF formula. If every clause in a CNF formula consists of exactly three distinct literals, then that formula is called a \textit{\threecnftext{} formula}. For example, $(x_1 \lor \overline{x_2} \lor x_3) \land (x_1 \lor x_2 \lor \overline{x_3})$ is a \threecnf{} formula.

\subsection{\threesattext{} (see \cite{hop-ull:b:automata})}
\textit{\threesattext} is the set of all \threecnf{} formulas such that there is some assignment of \true{} and \false{} to the variables of the formula that makes the formula evaluate to \true. When such an assignment exists, we say the formula is \textit{satisfiable}. For example, using the assignment $x_1 = x_2 = x_3 = \true$ makes the formula $(x_1 \lor \overline{x_2} \lor x_3) \land (x_1 \lor x_2 \lor \overline{x_3})$ evaluate to \true. Notice that there may be multiple such assignments. For example, $x_1 = \true,\ x_2 = x_3 = \false$ also makes the formula evaluate to \true.

\subsection{Resolution (see \cite{rus-nor:b:ai})}
A \textit{rule of inference} is a rule that, given Boolean expressions $A_1, \ldots, A_n$, derives a new expression $B$ that is \textit{entailed} by the set $\{A_1, \ldots, A_n\}$, i.e., the derived expression $B$ is true whenever $A_1$ through $A_n$ are all true. We denote such a rule by
\[ \frac{A_1,\, \ldots,\, A_n}{B}, \]
where the premises of the rule are written above the bar and the conclusion is written below.

Suppose that we have a \threecnf{} formula $E = \bigwedge_{i=1}^n C_i$, where $C_1, \ldots, C_n$ are the clauses of $E$. Further suppose that $E$ contains two clauses, $C_j$ and $C_k$, containing complementary literals, i.e., two literals such that one literal is the negation of the other. For example, if $C_j = a_1 \lor \ldots \lor a_l \lor c$ and $C_k = b_1 \lor \ldots \lor b_m \lor \overline{c}$, then $C_j$ and $C_k$ contain the complementary literals $c$ and $\overline{c}$. In this case, $C_j$ and $C_k$ entail a new clause $B$ as follows:
\[ \frac{a_1 \lor \ldots \lor a_l \lor c,\quad b_1 \lor \ldots \lor b_m \lor \overline{c}}{a_1 \lor \ldots \lor a_l \lor b_1 \lor \ldots \lor b_m}. \]
The above inference rule is called \textit{resolution}. If the resolvent $B$ is then added to the clauses of $E$, we obtain a new formula $E^\prime = B \land \left( \bigwedge_{i=1}^n C_i \right)$. Since $B$ is true whenever $C_j$ and $C_k$ are true, $E^\prime$ is satisfiable if and only if $E$ is satisfiable.

\section{Summary of Cho's Argument} \label{sec:summary}
Throughout their paper, there are multiple instances where Cho uses terms whose definitions are either informal or not present. Accordingly, we will not attempt to provide definitions for these terms. We will only describe their properties that can be inferred from usage and are necessary for our critique.

\subsection{Definitions}
\subsubsection{Indistinguishable Binomial Decision Trees}
Fundamental to Cho's argument is the \textit{indistinguishable binomial decision tree}, which we henceforth refer to simply as a ``decision tree.'' In this summary we detail only the essential aspects of the decision tree, and we refer the reader to Cho's paper~\cite{cho:indistinguishable} for the full construction.

A decision tree is a \threecnf{} formula with a special form. Such formulas are derived from a single unit clause---a clause consisting of exactly one literal---by a sequence of transformations which add variables and clauses to produce an arbitrarily large expression. The result of these transformations is a set of clauses which lends itself to interpretation as a graph. An example of such a set is provided in Figure~\ref{fig:tree}. Some clauses produced by the construction have been omitted; we reproduce only the essential structure.

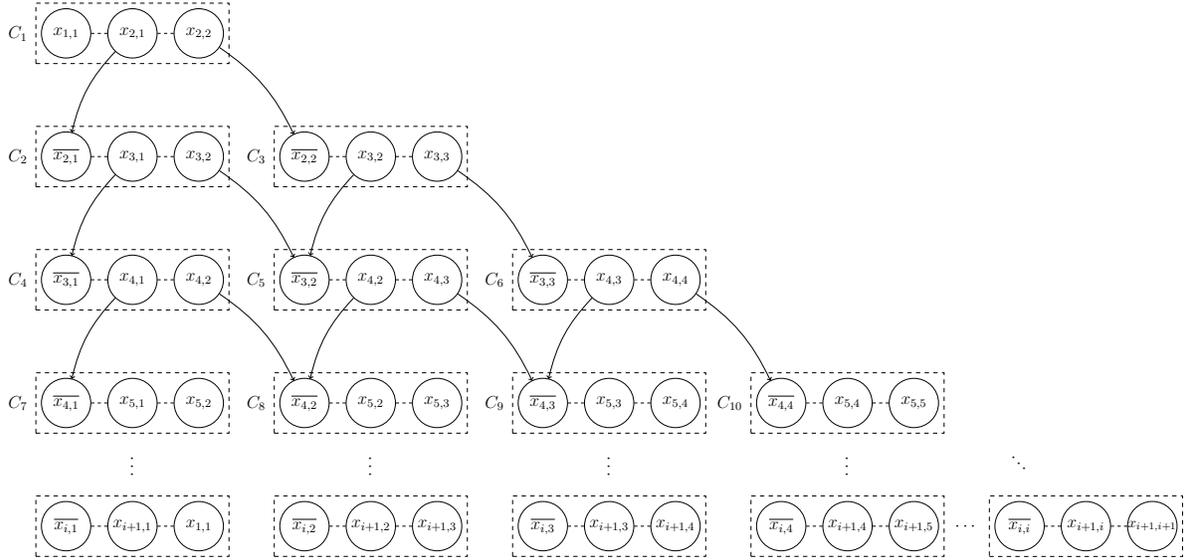
\begin{figure}[htb]
    \centering
    \begin{tikzpicture}[font=\large]

\clause{0,0}{1}{$C_1$}{$x_{1,1}$}{$x_{2,1}$}{$x_{2,2}$}

\foreach \r [remember=\r as \lastr (initially 1)] in {2,...,\explicitrows} {
    \begin{scope}[yshift=(\r-1)*-\rowsep]
        \level{0,0}{\r}
    \end{scope}
    \connectchildren{\lastr}
}

\clause{0,{(\rows-1)*-\rowsep}}{s1}{}{$\overline{x_{i,1}}$}{$x_{\plusone{i},1}$}{$x_{1,1}$}
\clause{{(\rows-1)*\clausesep},{(\rows-1)*-\rowsep}}{s\rows}{}{$\overline{x_{i,i}}$}{$x_{\plusone{i},i}$}{\small $x_{\plusone{i},\plusone{i}}$}

\foreach \col [count=\offset from 1, count=\colnxt from 3] in {2,...,\explicitrows} {
    \clause{\offset*\clausesep,{(\rows-1)*-\rowsep}}{s\col}{}{$\overline{x_{i,\col}}$}{$x_{\plusone{i},\col}$}{$x_{\plusone{i},\colnxt}$}
}

\foreach \i [count=\j from \explicitrows*(\explicitrows-1)/2+1] in {1,...,\explicitrows} {
    \path (\nodeid{\j}{2}) -- (\nodeid{s\i}{2}) node [pos=.5] {\rotatebox{90}{$\cdots$}};
}

\node[literal] (dummy) at (\explicitrows*\clausesep,{(\explicitrows-1)*-\rowsep}) {};

\path (dummy) -- (\nodeid{s\rows}{1}) node [pos=.5] {\rotatebox{-45}{$\cdots$}};
\path (\nodeid{s\explicitrows}{3}) -- (\nodeid{s\rows}{1}) node [pos=.5] {$\cdots$};

\end{tikzpicture}
    \caption{A set of clauses representing an indistinguishable binomial decision tree. Nodes of the induced graph are outlined with broken lines.}
    \label{fig:tree}
\end{figure}

A decision tree such as the one depicted in Figure~\ref{fig:tree} exhibits the structure described below. Let $E$ be the expression forming the decision tree. Each clause in $E$ is interpreted as being a node in a graph. The graph contains a directed edge from node $C_j$ to node $C_k$ if and only if $C_j$ contains a positive literal $x$ and $C_k$ contains its negation $\overline{x}$. We will see that this graph is a DAG\@, although Cho refers to it incorrectly as a tree.

The edges of the graph form what Cho terms a \textit{decision chain}. A decision chain is a list of pairs of complementary literals, which represent the resolved literals in a sequence of applications of the resolution inference rule. Thus the edges of the graph connect clauses containing complementary literals such that a path in the DAG corresponds to a valid sequence of resolution steps. For example, in Figure~\ref{fig:tree} there is a directed path connecting the nodes $C_1$, $C_2$, and $C_4$. As $C_1$ and $C_2$ contain the complementary literals $x_{2,1}$ and $\overline{x_{2,1}}$, they can be resolved to derive a new clause $B_1$. The clauses $B_1$ and $C_4$ now contain complementary literals $x_{3,1}$ and $\overline{x_{3,1}}$, and so the clauses can be resolved to derive a new clause $B_2$. Thus, the decision chain obtained along this path is the sequence $x_{2,1}, \overline{x_{2,1}}, x_{3,1}, \overline{x_{3,1}}$. The graph is constructed such that we may execute this process along any directed path.

In describing these paths, Cho introduces a notion of ``indistinguishability'' between variables. Furthermore, they claim that all paths from the single source node of the DAG to a given node are ``indistinguishable,'' since all such paths must have the same length and consist of indistinguishable variables. This notion, however, is never formally defined.

The structure discussed above and in Figure~\ref{fig:tree} characterizes an indistinguishable binomial decision tree. The decision tree is said to be ``binomial'' because the number of paths from the source node to a given successor node is always a binomial coefficient. This is shown in Lemma~3 of Cho's paper. The decision tree is said to be ``indistinguishable'' because all paths from the source to a given node are indistinguishable. This is shown in Lemma~2. However, as mentioned above, the notion of nodes being indistinguishable is never formalized.

\subsubsection{Entry Clauses}
Many of Cho's results are phrased in terms of ``entry clauses,'' but the term is never defined. At the very least, we can safely assume based on context that it does not refer to individual literals. This property alone is sufficient to demonstrate the flaw in Cho's argument, so henceforth we treat the term as a black box.

Cho also speaks of ``extracting'' entry clauses from a \threecnf{} instance, but they do not define this process either. In particular, Cho describes two different ways to ``extract'' entry clauses in the statements of Theorem~1 and Theorem~2, which are ``following the decision paths'' and ``following the tree levels.'' However, the specific details of these processes are not relevant to our critique, so we will not attempt to rigorously define these processes here.

\subsection{Cho's Argument}
We summarize the main results of Cho's paper below. Our intent is to enumerate each relevant claim made by Cho about the structures and algorithms used in their argument. Our critique will later identify a gap in Cho's argument where a nontrivial property of entry clauses is used without proof.

We have numbered the theorems below to coincide with the numbering of Cho's paper. The proofs of Lemma~1, Theorem~1, and Theorem~2 are omitted, as they are not relevant to our critique. We begin with Lemma~1, which pertains to the time needed to check the satisfiability of a decision tree. 

\renewcommand{\thelemma}{1}
\begin{lemma}[\cite{cho:indistinguishable}]
All variables contained in a binomial decision tree must be investigated in order to decide the satisfiability of an instance.
\end{lemma}

The lemma states that the number of steps needed to check the satisfiability of a decision tree is lower-bounded by the number of variables in the instance. We note that when Cho uses the term ``variables'' in the statement of Lemma~1, they are referring to the literals in the instance. Therefore, this amounts to showing that there is a linear lower bound on the runtime of any algorithm which decides the satisfiability of a decision tree. However, this is trivially known, as the runtime of any algorithm deciding \threesat{} has a linear lower bound.

Theorem~1 and Theorem~2 both relate to the time needed to ``extract'' entry clauses.

\renewcommand{\thetheorem}{1}
\begin{theorem}[\cite{cho:indistinguishable}]
All entry clauses contained in an indistinguishable binomial decision tree cannot be extracted in polynomial time following the decision paths.
\end{theorem}

\renewcommand{\thetheorem}{2}
\begin{theorem}[\cite{cho:indistinguishable}]
All entry clauses contained in an indistinguishable binomial decision tree cannot be extracted in polynomial time following the tree levels.
\end{theorem}

These two theorems purport to show that the processes of ``following the decision paths'' and ``following the tree levels'' both have superpolynomial lower bounds on their respective runtimes. Cho uses both of these results in their proof of Theorem~3 to argue that entry clauses cannot be extracted in polynomial time.

We treat Theorem~3 in greater depth. We include the statement of Theorem~3 and a summary of Cho's argument so that we may critique it in Section~\ref{sec:critique}.

\renewcommand{\thetheorem}{3}
\begin{theorem}[\cite{cho:indistinguishable}]
\threesat{} has no polynomial-time algorithm.
\end{theorem}

\begin{proof}[Summary of the purported proof:]
Cho starts by considering some \threecnf{} instance containing a decision tree. They cite Lemma~1 to conclude that every ``variable'' in the decision tree must be investigated to decide the satisfiability of the instance. We reiterate here that Cho uses the term ``variables'' to refer to literals in the decision tree instance. We also note that the result of Lemma~1 is not used again after it is initially mentioned. Next, Cho states that two ways to extract entry clauses are by ``following the decision paths'' or ``following the tree levels step by step.'' They then cite Theorem~1 and Theorem~2 to conclude that neither of these processes can be executed in polynomial time. From this, Cho concludes that entry clauses cannot be extracted from the instance in polynomial time. Finally, Cho concludes that the satisfiability of the instance cannot be determined in polynomial time by any algorithm.
\end{proof}

\section{Our Critique} \label{sec:critique}
Our critique is very simple. We point out that in Theorem~3, the theorem in which Cho claims that no polynomial-time algorithm can decide \threesat, Cho misuses the result of Lemma~1 and thus fails to correctly prove the theorem.

As described in our summary, Theorem~3 is proved by combining Lemma~1, Theorem~1, and Theorem~2. For the sake of argument, we will assume they are correct; we will focus only on the flaw in Theorem~3. We also want to reiterate that the term ``entry clause'' is never defined by Cho. We will treat the term as a black box, as a precise interpretation is not required to identify the flaw in Cho's argument.

In Theorem~3, Cho supposes that we are given a Boolean formula containing a decision tree and attempts to show that the satisfiability of such a formula cannot be determined in polynomial time. They first establish that every variable in the tree must be investigated in order to decide the satisfiability of the formula. The rest of the proof establishes that we cannot extract every entry clause in polynomial time, and thus satisfiability cannot be determined in polynomial time. The flaw in Cho's argument is that they fail to justify why every algorithm that decides \threesat{} needs to ``extract'' all entry clauses. We believe that Cho unintentionally interchanged ``variable'' and ``entry clause'' in this purported proof, since there is no argument connecting these two parts of the proof. To see this, recall from Section~\ref{sec:summary} that Cho uses the term ``variables'' to refer to literals, and we can safely assume that the term ``entry clause'' does not refer to individual literals. Thus ``entry clause'' is likely not an equivalent term for ``variables,'' and so Lemma~1 does not justify the necessity of extracting entry clauses. Likewise, Theorems 1 and~2 only establish a lower bound on the complexity of extracting all entry clauses. They do not justify why every entry clause must be extracted. Therefore, the superpolynomial lower bound established by Cho is only applicable to \emph{some} algorithms that decide \threesat, not all algorithms that decide \threesat. This is a trivial result. An example of such an algorithm is one that, after checking that the input is a valid Boolean formula, tries all possible variable assignments until it either finds one that satisfies the given formula, or has exhausted all assignments.

We believe a crucial flaw in Cho's approach is that it only considers resolution-based algorithms, which Cho deems important for proving that \threesat{} cannot be decided in polynomial time. Cho's work ignores prior work done by Haken~\cite{hak:a:resolution} and Ben-Sasson and Wigderson~\cite{ben-wig:a:short-proofs} which already establishes, correctly, a superpolynomial lower bound on resolution.

\section{Conclusion} \label{sec:conclusion}
We conclude that Cho's claimed proof that $\pe \subsetneq \np$ is not correct. Nonetheless, the approach used to construct hard Boolean formulas is quite interesting as it uses resolution in a creative manner---by reverse-engineering the process to obtain complex formulas from literals---but, unfortunately, it isn't applied correctly. This isn't Cho's first attempt to understand what makes \sat{} hard. In previous work, they analyzed hidden number systems within \sat{} to try to understand what makes some of its subsets, such as \hornsat{} and \twosat, tractable \cite{cho:number}. The search for an answer as to whether $\pe = \np$ might not always succeed, but it can potentially help enrich other areas of study and provide us with tools to analyze other problems and their structures, so we look forward to seeing more attempts being made.

\section{Acknowledgments} \label{sec:ack}
We thank Lane A. Hemaspaandra, Mandar Juvekar, Arian Nadjimzadah, David E. Narv\'{a}ez, and Patrick Phillips for their helpful feedback on earlier drafts of this paper. Any remaining errors are the responsibility of the authors.

\bibliography{ms}

\end{document}